\documentclass[11pt]{article}

\title{A Lower Bound for Nonadaptive, One-Sided Error Testing of Unateness of Boolean Functions over the Hypercube}

\author{Roksana Baleshzar\footnote{Department of Computer Science and Engineering, Pennsylvania State University. {\sf rxb5410@cse.psu.edu, ramesh@psu.edu, sofya@cse.psu.edu}. Partially supported by NSF award CCF-1422975. }
\and 
Deeparnab Chakrabarty\footnote{Department of Computer Science, Dartmouth College. {\sf deeparnab@dartmouth.edu}. Work done while at Microsoft Research, India.} 
\and 
Ramesh Krishnan S. Pallavoor\footnotemark[1] 
\and 
Sofya Raskhodnikova\footnotemark[1] 
\and 
C. Seshadhri\footnote{Department of Computer Science, University of California, Santa Cruz. {\sf sesh@ucsc.edu.}}
}

\date{}

\usepackage{paralist}
\usepackage{amsmath,amssymb,amsthm,mathtools}
\usepackage{xspace}
\usepackage{color}
\usepackage{xcolor}
\usepackage{fullpage}

\usepackage[colorlinks, citecolor = black,
urlcolor=black,linkcolor=black,pdfstartview=FitH]{hyperref}
\usepackage[ruled, noend, noline]{algorithm2e}
\usepackage{verbatim}
\usepackage{bbm}

\newtheorem{theorem}{Theorem}[section]

\newtheorem{claim}[theorem]{Claim}

\newtheorem{observation}[theorem]{Observation}

\newcommand{\eps}{\varepsilon}

\newcommand{\ord}[2][th]{\ensuremath{{#2}^{\mathrm{#1}}}}

\newcommand{\cD}{\mathcal{D}}

\begin{document}

\maketitle

\begin{abstract}
A Boolean function $f:\{0,1\}^d \mapsto \{0,1\}$ is unate if, along each coordinate, the function is either nondecreasing or nonincreasing. In this note, we prove that any nonadaptive, one-sided error unateness tester must make $\Omega(\frac{d}{\log d})$ queries. This result improves upon the $\Omega(\frac{d}{\log^2 d})$ lower bound for the same class of testers due to Chen et al.~(STOC, 2017).
\end{abstract}

\section{Introduction}
We study the problem of deciding whether a Boolean function $f:\{0,1\}^d \mapsto \{0,1\}$ is \emph{unate} in the property testing
 model~\cite{RS96,GGR98}. A function is unate if, for each dimension $i \in [d]$, the function is either nondecreasing along the $\ord{i}$ coordinate or nonincreasing along the $\ord{i}$ coordinate.
% 
% Unateness generalizes monotonicity and $\bb$-monotonicity.
% Monotone functions are monotone along all the dimensions, whereas $\bb$-monotone functions have a specific direction (i.e., $f$ is monotone along the $\ord{i}$ coordinate if $\bb_i = 0$ and antimonotone otherwise) along each dimension. 
% Note that a function is unate if it is $\bb$-monotone for some $\bb \in \{0,1\}^d$.
% 
% Property testing model was introduced by Rubinfeld and Sudan~\cite{RS96} and Goldreich et al.~\cite{GGR98}.
A property tester for unateness is a randomized algorithm that takes as input
a proximity parameter $\eps \in (0,1)$ and has query access to a function $f$.
If $f$ is unate, it must accept with probability at least $2/3$. If $f$ is $\eps$-far from
unate, it must reject with probability at least $2/3$.
A tester has \emph{one-sided error} if it always accepts unate functions. 
A tester is \emph{nonadaptive} if it chooses all of its queries in advance; it is \emph{adaptive} otherwise.

% 
% 
% A tester for a property $\cP$ of a function $f$ is an algorithm that takes as input a distance parameter $\eps \in (0,1)$ and has query access to $f$. It has to accept with probability at least $2/3$ if $f$ satisfies the property $\cP$ and reject with probability at least $2/3$ if $f$ is $\eps$-far from satisfying $\cP$. A function $f$ is $\eps$-far from $\cP$ if at least an $\eps$-fraction of values of $f$ has to be modified to make $f$ satisfy $\cP$. The domain $\{0,1\}^d$ is called the \emph{hypercube}.

The problem of testing unateness was introduced by Goldreich et al.~\cite{GGLRS00}. 
Following a result of Khot and Shinkar~\cite{KS16}, Baleshzar et al.~\cite{BCPRS17} settled the complexity
of unateness testing for \emph{real-valued functions}. 
Unateness can be tested with $O(\frac{d}{\eps})$ queries adaptively and with $O(\frac{d \log d}{\eps})$ queries nonadaptively.
For constant $\eps$, these complexities are optimal.
% They also show that these bounds are tight (for constant $\eps$) for real-valued functions.

On the other hand, for the Boolean range, the complexity is far from settled.
Baleshzar et al.~\cite{BMPR16} proved that $\Omega(\sqrt{d})$ queries are necessary for nonadaptive, one-sided error testers. Chen et al.~\cite{CWX17} improved the lower bound for this class of testers to $\Omega(\frac{d}{\log^2 d})$.
They also proved a lower bound of $\Omega(\frac{\sqrt{d}}{\log^2 d})$ for adaptive, two-sided error unateness testers.

In this note, we use a construction similar to the one used by Chen et al.~\cite{CWX17} to get an $\Omega(\frac{d}{\log d})$ for nonadaptive, one-sided error unateness testers of Boolean functions over the hypercube. 
Our analysis of the lower bound construction is simpler and gives a better dependence on $d$.
%In this note, we use the same construction used by Chen et al.~\cite{CWX17} to get an $\Omega(\frac{d}{\log d})$ lower bound using a different method of analysis.
There is still a gap of $\log^2d$ between the query complexity of the best known algorithm for this problem (from~\cite{BCPRS17}) and our lower bound.

\section{The Lower Bound}
In this section, we prove the following theorem.
%a lower bound of $\Omega(d/\log d)$ for nonadaptive, one-sided error unateness testers of Boolean functions $f:\{0,1\}^d \to \{0,1\}$.
\begin{theorem}\label{thm:main-lb}
Any nonadaptive, one-sided error unateness tester for functions $f:\{0,1\}^d \mapsto \{0,1\}$ with the distance parameter $\eps \leq \frac{1}{8}$ must make $\Omega(\frac{d}{\log d})$ queries.
\end{theorem}
\begin{proof}
We first define a hard distribution consisting of Boolean functions that are $\frac{1}{8}$-far from unate.
By Yao's minimax principle~\cite{Yao77}, it is sufficient to give a distribution on functions for which every deterministic tester fails with high probability.
A deterministic nonadaptive tester is determined by a set of query points $Q \subseteq \{0,1\}^d$.
We prove that if $|Q| \leq \frac{d}{30 \log d}$, then the tester fails with probability more than $2/3$ over the
hard distribution. 

%Consider unateness tester for Boolean functions $f:\{0,1\}^d \to \{0,1\}$. 
The hard distribution $\cD$ is defined as follows: pick $3$ dimensions $a,b,c \in [d]$ uniformly at random and define $f_{a,b,c}(x) = x_a \cdot x_b + (1-x_a) \cdot x_c$. We call $a,b,c$ the {\em influential dimensions}, since the value of the function depends only on them. The coordinate $x_a$ determines if $f_{a,b,c}(x)$ should be set to $x_b$ or $x_c$. If $x_a = 1$, then $f_{a,b,c}(x) = x_b$, otherwise, $f_{a,b,c}(x) = x_c$.

There are $d \choose 3$ functions in the support of $\cD$.
The next claim states that all of them are far from unate.
%There are $d \choose 3$ functions in $\cD$.
\begin{claim}
Every function $f_{a,b,c}$ in the support of $\cD$ is $\frac{1}{8}$-far from unate.
\end{claim}
\begin{proof}
Consider an edge $(x,y)$ along the dimension $a$. We have $x_a = 0$ and $y_a = 1$, and $x_i = y_i$ for all $i \in [d] \setminus \{a\}$. 
By definition, $f_{a,b,c}(x) = x_c$ and $f_{a,b,c}(y) = y_b$.
If $x_b = y_b = 1$ and $x_c = y_c = 0$, then $f_{a,b,c}$ is increasing along the edge $(x,y)$. 
On the other hand, if $x_b = y_b = 0$ and $x_c = y_c = 1$, then $f_{a,b,c}$ is decreasing along $(x,y)$. Thus, with respect to $f_{a,b,c}$, at least $2^{d-3}$ edges along the dimension $a$ are decreasing and at least $2^{d-3}$ edges along the dimension $a$ are increasing. 
Hence, at least $2^{d-3}$ function values of $f_{a,b,c}$ need to be changed to make it unate. 
Consequently, $f_{a,b,c}$ is $\frac{1}{8}$-far from unate.
\end{proof}

\noindent Note that any one-sided error tester for unateness must accept if the query answers are consistent with a unate function.
Let $f_{|Q}$ denote the restriction of the function $f$ to the points in $Q$. 
%If there exists a unate function $g$ such that $g_{|Q} = f_{|Q}$, then the tester must accept.
We say that $f_{|Q}$ is {\em extendable} to a unate function if there exists a unate function $g$ such that $g_{|Q} = f_{|Q}$.
For $f \sim \cD$, we show that if $|Q| \leq \frac{d}{30 \log d}$, then, with high probability, $f_{|Q}$ is extendable to a unate function. Consequently, the tester accepts with high probability.

Next, we define a conjunctive normal form (CNF) formula $\phi(f_{|Q})$.
Intuitively, each pair $(x,y)$ of domain points on which $f$ differs imposes a constraint on $f$ (assuming that $f$ is unate).
Specifically, at least one of the dimensions on which $x$ and $y$ differ must be consistent (i.e., nondecreasing or nonincreasing) with the change of the function value between $x$ and $y$.
This constraint is formalized in the definition of $\phi(f_{|Q})$ as follows.
For each dimension $i$, we have a variable $z_i$ 
which is true if $f$ is nondecreasing along the dimension $i$, and false if it is nonincreasing along that dimension.
For each $x, y \in Q$ such that $f(x) = 1$ and $f(y) = 0$, create a clause (think of $x,y$ as sets where $i \in x$ iff $x_i = 1$)
$$ c_{x,y} = \bigvee_{i \in x\setminus y} z_i \vee \bigvee_{i \in y\setminus x} \overline{z_i} .$$
Set $\phi(f_{|Q}) = \bigwedge_{x,y \in Q: f(x) = 1, f(y) = 0} c_{x,y}$.

\begin{observation} 
The restriction $f_{|Q}$ is a certificate for non-unateness iff $\phi(f_{|Q})$ is unsatisfiable.
\end{observation}

Now we need to show that, with probability greater than $2/3$ over $f \sim \cD$, the CNF formula $\phi(f_{|Q})$ is satisfiable.
This follows from Claims~\ref{clm:width-of-cnf} and~\ref{clm:cnf-satisfiable}.

The width of a clause is the number of literals in it; the width of a CNF formula is the minimum width of a clause in it.

\begin{claim} \label{clm:width-of-cnf}
With probability at least $2/3$ over $f \sim \cD$, the width of
$\phi(f_{|Q})$ is at least $3\log d$.
\end{claim}

\begin{proof} 
Consider a graph $G$ with vertex set $Q$, and an edge between $x,y \in Q$
if $|x\Delta y| \leq 3\log d$ (Here, $x \Delta y$ is the symmetric difference between the sets $x$ and $y$). 
Take an arbitrary spanning forest $F$ of $G$.
Observe that for any edge $(u,v)$ of $G$, we have $u\Delta v \subseteq \bigcup_{(x,y) \in F} x\Delta y$.
Note that $F$ has at most $\frac{d}{30\log d}$ edges. 
Let $C = \bigcup_{(x,y) \in F} x \Delta y$, the set of dimensions captured by $Q$.
%The total number of dimensions captured by $Q$ is
We have $|C| \leq \sum_{(x,y) \in F} |x\Delta y|
\leq \frac{d}{30 \log d} \cdot 3\log d \leq \frac{d}{10}$. 
Over the distribution $\cD$, the probability that at least one of the influential dimensions, $\{a,b,c\}$, is in $C$ is at most ${3}/{10}$ which is less than $1/3$.
Hence, with probability at least ${2}/{3}$,
no $(u,v) \in G$ contributes a clause to $\phi(f_{|Q})$. Therefore, the width
of $\phi(f_{|Q})$ is at least $3\log d$.
\end{proof} 

\begin{claim}  \label{clm:cnf-satisfiable}
Any CNF that has width at least $3\log d$ and at most $d^2$ clauses is satisfiable.
\end{claim}

\begin{proof} 
Apply the probabilistic method. A clause is not satisfied by a random assignment with
probability at most $1/d^3$.
Hence, the expected number of unsatisfied clauses is at most $\frac{d^2}{d^3} < 1$.
\end{proof}

\noindent
Thus, $f_{|Q}$ is a certificate for non-unateness with probability at most $1/3$ when $|Q| \leq \frac{d}{30 \log d}$, which completes the proof of Theorem~\ref{thm:main-lb}.
\end{proof}

\bibliography{references}

\end{document}